\documentclass[11pt]{amsart}
\usepackage{amssymb,latexsym}
\input xypic
  \xyoption{all}

\usepackage{amsmath, amsthm, amsfonts}

\usepackage{pifont,pxfonts,txfonts}
\usepackage{yfonts}
\usepackage{bbm}
\usepackage{calrsfs}
\usepackage{empheq}
\usepackage{color}
\usepackage[colorlinks,linkcolor=red,anchorcolor=blue,citecolor=green]{hyperref}
\usepackage{amsmath, amstext, amsbsy, amscd}
\usepackage[mathscr]{eucal}
\usepackage{times}
\usepackage{amsmath, amsthm, amsfonts,extarrows}

\newtheorem{theorem}{Theorem}[section]
\newtheorem{proposition}[theorem]{Proposition}
\newtheorem{lemma}[theorem]{Lemma}

\newtheorem{conjecture}[theorem]{Conjecture}

   \newcommand{\ba}{\begin{eqnarray}}
   \newcommand{\na}{\end{eqnarray}}
   \newcommand{\ban}{\begin{eqnarray*}}
   \newcommand{\nan}{\end{eqnarray*}}

\newcommand{\bC}{{\mathbb C}}

\newcommand{\bQ}{{\mathbb Q}}

  \newcommand{\<}{\langle}
  \renewcommand{\>}{\rangle}

\newcommand{\suml}{\sum\limits}

\begin{document}

\title{On a conjectural solution to open KdV and Virasoro}

\author[Hua-Zhong Ke]{Hua-Zhong Ke}
  \address{Department of Mathematics\\  Sun Yat-Sen University\\
                        Guangzhou,  510275\\ China }
  \email{kehuazh@mail.sysu.edu.cn}

\begin{abstract}

In this note, we present a recursive formula for the full partition function $Z$ of descendent integrals over moduli spaces of open and closed Riemann surfaces, assuming the conjecture recently proposed by Pandharipande, Solomon and Tessler that $Z$ satisfies the open KdV and Virasoro equations.


\end{abstract}

\date{\today}
\maketitle
\tableofcontents

\section{Introduction}

The study of intersection theory on moduli spaces of closed Riemann surfaces has a long history. Let $\overline M_{g,l}$ be the Deligne-Mumford moduli space of stable pointed curves of genus $g$ with $l$ distinct nonsingular marked points, which has (complex) dimension $3g-3+l$. Note that it is nonempty if and only if $2g-2+l>0$ (stability condition). Let $\psi_i$ be the first Chern class of the cotangent line bundle at the $i$-th marked point. The descendent integrals are defined by
\ban
\<\tau_{a_1}\cdots\tau_{a_l}\>_g^c=\int_{\overline M_{g,l}}\psi_1^{a_1}\cdots\psi_l^{a_l}.
\nan
Consider the generating series
\ban
F^c(u;t_0,t_1,\cdots)=\sum\limits_{g=0}^\infty u^{2g-2}\<\exp(\sum\limits_{i=0}^\infty t_i\tau_i)\>_g^c.
\nan
Witten conjectured that the partition function $Z_{WK}=\exp(F^c)$ is a tau-function of the KdV hierarchy \cite{W}, which was first proved by Kontsevich \cite{K}. Dijkgraaf, Verlinde and Verlinde showed that Witten's conjecture is equivalent to Virasoro constraints \cite{DVV}. Using Virasoro constraints, Alexandrov found a recursive formula for $Z_{WK}$ in terms of a cut-and-join type differential operator \cite{A}, by assigning suitable degrees to the variables in a clever way. Inspired by this result, Zhou was able to obtain an explicit formula for $Z_{WK}$ \cite{Z}.

The study of intersection theory on moduli spaces of open Riemann surfaces was initiated only recently by Pandharipande, Solomon and Tessler  \cite{PST}. They have constructed the compactification of moduli spaces of open Riemann surfaces in the case $g=0$, and defined a theory of open descendent integration for pointed disks, which was also solved completely. Although the full theory of open descendent integration is not constructed yet, Pandharipande et al. conjectured that the full partition function $Z$ satisfies two different types of equations, which are open analog of KdV and Virasoro equations. It is important to determine geometrically interesting solutions to these equations. Buryak has shown that these two types of equations are equivalent \cite{B}. 

In this note, assuming that $Z$ satisfies open Virasoro equaitions, we follow the idea of \cite{A} to obtain a recursive formula for $Z$ (Proposition \ref{recursiveformula}). We will investigate the explicit formula for $Z$ in the future.

\section{Moduli of open Riemann surfaces}

In this section, we briefly describe some basic materials on the moduli spaces of open Riemann surfaces. We refer readers to \cite{PST} for more details.

Let $\Delta\subset\bC$ be the open unit disk, and let $\bar\Delta$ be its closure. An extendable embedding of $\Delta$ in a closed Riemann surface $C$ is a holomorphic embedding $\Delta\hookrightarrow C$ which extends to a holomorphic embedding of an open neighborhood of $\bar\Delta$ in $C$. Two extendable embeddings in $C$ are said to be disjoint if the images of $\bar\Delta$ are disjoint.

An open Riemann surface is a Riemann surface with boundary $(X,\partial X)$, which is obtained by removing finitely many disjoint extendably embedded open disks from a connected closed Riemann surface. The boundary $\partial X$ is the union of images of the unit circle boundaries of embedded disks.

Given an open Riemann surface $(X,\partial X)$, we can canonically construct the double $D(X,\partial X)$ via Schwarz reflection through the boundary, which is a connected closed Riemann surface. The double genus of $(X,\partial X)$ is defined to be the usual genus of $D(X,\partial X)$.

On an open Riemann surface $(X,\partial X)$, we consider two types of marked points. The marked points of interior type are points of $X\setminus\partial X$, and the marked points of boundary type are points of $\partial X$.

Let $M_{g,k,l}$ be the moduli space of open Riemann surfaces of doubled genus $g$ with $k$ distinct boundary marked points and $l$ distinct interior marked points, which is a real orbifold with real dimension $3g-3+k+2l$. The moduli space is not compact, and may be disconnected and non-orientable. Note that it is nonempty if and only if $2g-2+k+2l>0$ (stability condition).

Note that the cotangent line bundles on $M_{g,k,l}$ at interior marked points are well-defined. The authors of \cite{PST} do not consider contangent line bundles at boundary marked points. Naively, open descendent integrals are expected to be defined by
\ban
\<\tau_{a_1}\cdots\tau_{a_l}\sigma^k\>_g^o=\int_{\overline M_{g,k,l}}\psi_1^{a_1}\cdots\psi_l^{a_l},
\nan
where $\overline M_{g,k,l}$ is a suitable compactification of $M_{g,k,l}$, and $\psi_i$ is the first Chern class of the cotangent line bundle at the $i$-th interior marked point. 

To define the open descendent integrals rigorously, at least three significant steps must be taken:
\begin{itemize}
\item A natural compactification $M_{g,k,l}\subset\overline M_{g,k,l}$ must be constructed, since $M_{g,k,l}$ is not compact.
\item For integration over $\overline M_{g,k,l}$ to be well-defined, boundary conditions of the integrands along $\partial\overline M_{g,k,l}$ must be specified.
\item Orientation problems for $\overline M_{g,k,l}$ must be solved, since $M_{g,k,l}$ may be non-orientable.
\end{itemize}
The authors of \cite{PST} have completed the above steps in the case $g=0$ and obtained a complete description of open descendent integrals for pointed disks. 

Consider the generating series
\ban
F^o(u;s;t_0,t_1,\cdots)&=&\sum\limits_{g=0}^\infty u^{g-1}\<\exp(s\sigma+\suml_{d=0}^\infty t_d\tau_d)\>_g^o.
\nan
Though the full theory of open descendent integration is not constructed yet, the authors of \cite{PST} conjectured that the full partition function $Z=\exp(F^c+F^o)$ satisfies open KdV and Virasoro equations. 

\section{Derivation of the recursive formula}

The Virasoro constraint for $Z_{WK}$ is as follows:
\ban
\frac{(2n+3)!!}{2^{n+1}}\partial_{n+1}Z_{WK}=\hat{L}_nZ_{WK},\quad n\geqslant -1,
\nan
where
\ban
\hat{L}_n&=&\sum\limits_{m=0}^\infty\frac{\Gamma(m+n+\frac{3}{2})}{\Gamma(m+\frac{1}{2})}t_m\partial_{m+n}+\frac{u^2}{2^{n+2}}\sum\limits_{a+b=n-1}(2a+1)!!(2b+1)!!\partial_a\partial_b\\
&&\quad+\frac{1}{2u^2}\delta_{n,-1}t_0^2+\frac{1}{16}\delta_{n,0},
\nan
and $\partial_a=\frac{\partial}{\partial t_a}$. $Z_{WK}$ is uniquely determined by the Virasoro equations and the initial data $Z_{WK}|_{t_0=t_1=\cdots=0}=1$.

$s$ extension of the above constraints gives the open Virasoro constraint for $Z$. 
\begin{conjecture}\label{virasoro}(Conjecture 1 in \cite{PST})
The full partition function $Z=\exp(F^c+F^o)$ satisfies the following equations:
\ban
\frac{(2n+3)!!}{2^{n+1}}\partial_{n+1}Z=\hat{\mathscr L}_nZ,\quad n\geqslant -1,
\nan
where each $\hat{\mathscr L}_n$ is an $s$ extension of $L_n$:
\ban
\hat{\mathscr L}_n=\hat L_n+u^ns\frac{\partial^{n+1}}{\partial s^{n+1}}+\frac{3n+3}{4}u^n\frac{\partial^{n}}{\partial s^{n}}.
\nan
\end{conjecture}

Assuming Conjecture \ref{virasoro}, $Z$ is uniquely determined by the open Virasoro equations and the initial data $Z|_{t_0=t_1=\cdots=0}$. Note that by the stability conditions and dimension constraints, we have
\ban
F^c|_{t_0=t_1=\cdots=0}&=&0,\\
F^o|_{t_0=t_1=\cdots=0}&=&\frac{\<\sigma^3\>_0^o}{6} u^{-1}s^3,
\nan
where $\<\sigma^3\>_0^o=1$ as computed in \cite{PST}. This gives the initial data
\ba\label{initialdata}
Z|_{t_0=t_1=\cdots=0}=\sum\limits_{m=0}^\infty\frac{u^{-k}s^{3k}}{k!6^k}.
\na

Set deg$u=0$ and deg$t_n=\frac{2n+1}{3}$ as in \cite{A}, and set deg$s=\frac 23$.  Then $\hat{\mathscr L}_n$ is homogeneous of degree $-\frac{2n}{3}$. Moreover, note that $\<\tau_{a_1}\cdots\tau_{a_l}\>_g^c\neq0$ only if
\ban
a_1+\cdots+a_l=3g-3+l,
\nan
which gives
\ban
\frac{2a_1+1}{3}+\cdots+\frac{2a_l+1}{3}=2g-2+l,
\nan
and $\<\tau_{a_1}\cdots\tau_{a_l}\sigma^k\>_g^o\neq0$ only if 
\ban
2a_1+\cdots+2a_l=3g-3+k+2l,
\nan
which gives
\ban
\frac{2a_1+1}{3}+\cdots+\frac{2a_l+1}{3}+\frac{2k}{3}=g-1+k+l.
\nan
So $F^c,F^o\in\bQ[u^{-1},u][[s,t_0,t_1,\cdots]]$ have positive integral degrees by the stability condition. The low degree terms can be computed recursively:
\ban
F^c&=&(u^{-2}\frac{t_0^3}{6}+\frac{t_1}{24})+(u^{-2}\frac{t_0^3t_1}{6}+\frac{t_0t_2}{24}+\frac{t_1^2}{48})+\textrm{higher degree terms},\\
F^o&=&(u^{-1}st_0+\frac{t_1}{2})+(u^{-1}\frac{s^3}{6}+u^{-1}st_0t_1+\frac{t_0t_2}{2}+\frac{t_1^2}{4})+\textrm{higher degree terms}.
\nan
We also have a natural grading for $Z\in\bQ[u^{-1},u][[s,t_0,t_1,\cdots]]$:
\ban
Z=\suml_{d=0}^\infty Z_d\textrm{ with }Z_0=1.
\nan

Let $E=\frac 23s\frac{\partial}{\partial s}+\sum\limits_{n=0}^\infty\frac{2n+1}{3}t_n\partial_n$ be the Euler vector field. Then we have
\ban
EZ_d=dZ_d,\quad d\geqslant 0.
\nan
Set
\ban
\mathscr W=\sum\limits_{n=0}^\infty\frac{2n+1}{3}t_n\cdot\frac{2^n}{(2n+1)!!}\hat{\mathscr L}_{n-1}.
\nan
Then $\mathscr W$ is a homogeneous operator of degree one, and Conjecture \ref{virasoro} gives
\ban
\mathscr WZ=(E-\frac 23s\frac{\partial}{\partial s})Z\Rightarrow\mathscr WZ_d=(d+1-\frac 23s\frac{\partial}{\partial s})Z_{d+1}.
\nan
For $d\geqslant 0$, define operators $\mathscr O_{d+1}$ as follows:
\ban
\mathscr O_{d+1}(\sum\limits_{k=0}^\infty c_ks^k)=\left\{\begin{array}{cc}\suml_{k=0}\frac{c_k}{d+1-\frac 23 k}s^k,&d\textrm{ even,}\\\\\frac{u^{-\frac{d+1}{2}}s^{\frac 32 (d+1)}}{(\frac{d+1}{2})!\cdot 6^{\frac{d+1}{2}}}+\suml_{\substack{0\leqslant k\leqslant\infty\\k\neq\frac 32 (d+1)}}\frac{c_k}{d+1-\frac 23 k}s^k,&d\textrm{ odd.}\end{array}\right.
\nan
\begin{lemma}
\ban
\mathscr O_{d+1}\mathscr WZ_d=Z_{d+1}.
\nan
\end{lemma}
\begin{proof}
The case $d$ even is clear. When $d$ is odd, write
\ban
Z_{d+1}=\suml_{k=0}^\infty c_ks^k.
\nan
Then 
\ban
\mathscr WZ_d=\suml_{\substack{0\leqslant k\leqslant\infty\\k\neq\frac 32(d+1)}}(d+1-\frac 23k)c_ks^k.
\nan
By degree constraint, $c_{\frac 32(d+1)}$ does not contain $t_0,t_1,\cdots$. So from the initial data \eqref{initialdata}, we have 
\ban
c_{\frac 32(d+1)}s^{\frac 32(d+1)}=\frac{u^{-\frac{d+1}{2}}s^{\frac 32 (d+1)}}{(\frac{d+1}{2})!\cdot 6^{\frac{d+1}{2}}},
\nan
which proves the lemma.
\end{proof}

From the above lemma, we obtain the following recursive formula:
\begin{proposition}\label{recursiveformula}
\ban
Z=1+\suml_{d=1}^\infty(\mathscr O_d\mathscr W)\cdots(\mathscr O_1\mathscr W)1.
\nan
\end{proposition}
We point out that we need the initial data \eqref{initialdata} to construct operators $\mathscr O_d$.

{\bf Acknowledgements.}

The author would like to thank Prof. Jian Zhou for helpful discussions.

\end{document}